\documentclass[sigconf]{acmart}
\AtBeginDocument{%
  }
\usepackage{fancyhdr}
\usepackage{balance} 
\copyrightyear{2026}
\acmYear{2026}
\setcopyright{cc}
\setcctype{by}
\acmConference[WWW '26]{Proceedings of the ACM Web Conference 2026}{April 13--17, 2026}{Dubai, United Arab Emirates}
\acmBooktitle{Proceedings of the ACM Web Conference 2026 (WWW '26), April 13--17, 2026, Dubai, United Arab Emirates}
\acmDOI{10.1145/3774904.3792182}
\acmISBN{979-8-4007-2307-0/2026/04}

\acmSubmissionID{rfp0806}

\usepackage{algorithm}
\usepackage{algorithmicx}
\usepackage{subcaption}
\usepackage{multirow}
\usepackage{booktabs} 
\usepackage{tabularx}
\usepackage{url}
\usepackage{wrapfig}
\usepackage[noend]{algpseudocode}
\usepackage{booktabs}
\newenvironment{icompact}{
	\begin{list}{$\bullet$}{
			\itemindent -.05em
			\parsep 0pt plus 1pt
			\partopsep 0pt plus 1pt
			\topsep 2pt plus 2pt minus 2pt
			\itemsep 0pt plus 1.3pt
			\parskip 0pt plus 2pt
			\leftmargin 0.13in}
	}
	{\normalsize
	\end{list}
}
\newcommand{\para}[1]{\vspace{2pt}\noindent{\textbf{#1}}\hspace{10pt}\vspace{0.1pt}}

\begin{document}\sloppy

\title{Reconstructing Training Data from Adapter-based Federated Large Language Models}

\settopmatter{authorsperrow=4}

\author{Silong Chen}
\orcid{0009-0003-3904-3374}
\affiliation{%
	\institution{National University of Defense Technology}
	\city{Changsha}
	\country{China}
}
\email{chensilong@nudt.edu.cn}

\author{Yuchuan Luo}
\orcid{0000-0002-0720-4925}
\affiliation{%
	\institution{National University of Defense Technology}
	\city{Changsha}
	\country{China}
}
\email{luoyuchuan09@nudt.edu.cn}
\authornote{Yuchuan Luo and Yi Liu are co-corresponding authors.}

\author{Guilin Deng}
\orcid{0000-0001-8352-9759}
\affiliation{%
	\institution{National University of Defense Technology}
	\city{Changsha}
	\country{China}
}
\email{dengguilin@nudt.edu.cn}

\author{Yi Liu}
\orcid{0000-0002-0811-6150}
\affiliation{%
	\institution{City University of Hong Kong}
	\city{Hong Kong}
	\country{China}
}
\email{yiliu247-c@my.cityu.edu.hk}
\authornotemark[1]

\author{Ming Xu}
\orcid{0000-0002-2657-5764}
\affiliation{%
	\institution{National University of Defense Technology}
	\city{Changsha}
	\country{China}
}
\email{xuming@nudt.edu.cn}

\author{Shaojing Fu}
\orcid{0000-0002-7275-8190}
\affiliation{%
	\institution{National University of Defense Technology}
	\city{Changsha}
	\country{China}
}
\email{fushaojing@nudt.edu.cn}

\author{Xiaohua Jia}
\orcid{0000-0001-8702-8302}
\affiliation{%
  \institution{City University of Hong Kong}
  \city{Hong Kong}
  \country{China}
}
\email{csjia@cityu.edu.hk}

\renewcommand{\shortauthors}{Silong Chen et al.}

\begin{abstract}
Adapter-based Federated Large Language Models (FedLLMs) are widely adopted to reduce the computational, storage, and communication overhead of full-parameter fine-tuning for web-scale applications while preserving user privacy. By freezing the backbone and training only compact low-rank adapters, these methods appear to limit gradient leakage and thwart existing Gradient Inversion Attacks (GIAs).

Contrary to this assumption, we show that low-rank adapters create new, exploitable leakage channels. We propose the Unordered-word-bag-based Text Reconstruction (UTR) attack, a novel GIA tailored to the unique structure of adapter-based FedLLMs. UTR overcomes three core challenges—low-dimensional gradients, frozen backbones, and combinatorially large reconstruction spaces—by:
(i) inferring token presence from attention patterns in frozen layers,
(ii) performing sentence-level inversion within the low-rank subspace of adapter gradients, and
(iii) enforcing semantic coherence through constrained greedy decoding guided by language priors. Extensive experiments across diverse models (GPT2-Large, BERT, Qwen2.5-7B) and datasets (CoLA, SST-2, Rotten Tomatoes) demonstrate that UTR achieves near-perfect reconstruction accuracy (ROUGE-1/2 > 99), even with large batch sizes—settings where prior GIAs fail completely. Our results reveal a fundamental tension between parameter efficiency and privacy in FedLLMs, challenging the prevailing belief that lightweight adaptation inherently enhances security. Our code and data are available at \href{https://github.com/shwksnshwowk-wq/GIA}{link}.
\end{abstract}

\ccsdesc[500]{Security and privacy}
\ccsdesc[300]{Computing methodologies~Machine learning}

\keywords{Federated Learning, Large Language Models, Gradient Inversion Attacks, Federated LLMs}


\maketitle

\section{Introduction}
In modern web applications, such as personalized web agents~\cite{cai2025large,liu2025fedmobile}, intelligent chatbots~\cite{liu2020towards}, and collaborative content generation platforms~\cite{wang2024content}, there is an increasing demand for data and model collaboration across multiple users or devices to enhance performance and user experience. However, such collaboration often requires access to and processing of sensitive user private information, including personal messages, browsing histories, and transaction records, thereby raising serious privacy concerns~\cite{FL}. Direct data sharing poses significant risks of data leakage or misuse, particularly under stringent data protection regulations such as the GDPR~\cite{voigt2017eu} and CCPA~\cite{bonta2022california}. As a result, developing privacy-preserving techniques for collaborative learning has become an important research challenge in both academia and industry \cite{Weather_Forecasting, biomedical_NLP}.

Federated Large Language Model (FedLLM)~\cite{ye2024openfedllm,FedLLMs_ijcai2024p919,kuang2024federatedscope} framework (e.g., Federatedscope-LLM~\cite{kuang2024federatedscope}) has emerged as a promising solution to privacy concerns in collaborative learning by combining the powerful capabilities of LLMs with the privacy-preserving principles of FL. This paradigm enables decentralized model training without the need for centralized data collection~\cite{FL}. In particular, multiple clients collaboratively train a shared global model by exchanging only local model updates (e.g., gradients or parameters) while keeping their raw data securely stored on local devices~\cite{liu2022right}. Given that LLMs typically contain millions or even billions of parameters, adapter-based FedLLM frameworks~\cite{long2024dual,wang2024flora,bai2024federated,cai2023efficient} have become a mainstream approach within FedLLM systems to enhance training efficiency and reduce communication overhead. Adapters are lightweight, parameter-efficient modules (e.g., Low-Rank Adaptation (LoRA)~\cite{hu2022lora}) that can be seamlessly integrated into pre-trained LLMs, enabling fine-tuning for specific tasks with minimal computational and communication costs. By updating only a small subset of parameters, adapter-based FedLLMs achieve efficient model adaptation without compromising performance, making them particularly suitable for resource-constrained environments and expanding their applicability across diverse domains.

Despite these advantages, the security and privacy guarantees of adapter-based FedLLMs\footnote{To avoid ambiguity in terminology, the terms ``adapter-based FedLLMs” and ``FedLLMs” are used interchangeably throughout the remainder of this paper.} may not be as robust as they appear. We argue this for two primary reasons. First, gradient inversion attacks~\cite{data_reconstruction_attack,dager,Gradient_Inversion_Attack,recovering} are a well-established threat that can compromise data privacy by recovering sensitive training data from shared gradients. For instance, Zhu \textit{et al.}~\cite{zhu2019deep} demonstrated the Deep Gradient Leakage (DGL) attack against traditional FL systems, revealing that adversaries can reconstruct private client data from exchanged gradients, posing one of the most severe privacy risks in collaborative learning. Second, similar to standard FL, adapter-based FedLLMs rely on the exchange of model parameters or gradients during training, inherently creating an attack surface for malicious participants or eavesdroppers. While such attacks are generally more challenging to execute in practical FedLLM settings, since existing methods primarily target full-parameter fine-tuning settings~\cite{hugging_face_adapter,bai2024federated,wu2024cardinality,guo2023pfedprompt,LAMP,TAG} where gradients carry richer informational content, the risk is not entirely mitigated. Adapter-based FedLLMs update only a small subset of parameters, which can obscure certain data features and reduce the attack surface. Nevertheless, adaptive adversaries may still refine their techniques to exploit even sparse gradient information, indicating that privacy leakage remains a tangible and evolving threat in this paradigm.

In this paper, we disclose that adapter-based FedLLM frameworks remain vulnerable to serious privacy attacks and propose the Unordered-word-bag-based Text Reconstruction (UTR) attack, which is a new attack designed to recover sensitive training text. UTR pursues two primary objectives: (1) accurately reconstruct private textual training examples, and (2) evade common defensive mechanisms. We further demonstrate that UTR transfers effectively across a range of FedLLM configurations, including different adapter designs and underlying LLM backbones. Achieving such a powerful attack is non-trivial; below we summarize the three principal challenges we confronted:
\begin{icompact}
	\item \textbf{C1: Lower Dimensionality of Updates.} Adapter-based frameworks update only a small fraction of parameters, drastically reducing gradient dimensionality and limiting raw information leakage. Consequently, traditional optimization-based gradient inversion attacks (e.g., DLG~\cite{zhu2019deep}), which depend on rich gradient signals, fail to reconstruct inputs accurately~\cite{TAG,dager}.
	
	\item \textbf{C2: Frozen Backbone.} In adapter-based FedLLM, the backbone of the pre-trained LLM remains frozen, meaning its parameters are not updated during training~\cite{hu2022lora}. Consequently, attackers can only access gradients from the adapter layers, while the key input–output mappings embedded in the attention and transformer layers of the backbone remain hidden. Even state-of-the-art methods such as DAGER~\cite{dager}, which can accurately reconstruct text from attention layers in full-parameter tuning, suffer reduced inversion accuracy when such layers are inaccessible.
	
	\item \textbf{C3: Extensive Solution Space.} The solution space of possible reconstructions expands dramatically due to the large vector space induced by the embedding adapters. Direct traversal or beam search, as used in prior gradient inversion approaches~\cite{LAMP,TAG,dager}, becomes computationally infeasible when the batch size is large (e.g., $\geqslant 32$).
	
\end{icompact}

To address C1, UTR abandons traditional gradient vector fitting methods (e.g., distance-based optimization) and instead leverages the frozen attention layers of the LLM to directly infer token presence from the vocabulary, determining whether a specific token was used during training. To address C2, we observe that adapter layers in frozen LLMs exhibit low-rank properties; UTR exploits the gradient-generated subspace of this structure to estimate the degree of sentence inversion, thereby enhancing reconstruction accuracy. To address C3, UTR incorporates grammatical, semantic, and contextual continuity constraints to drastically reduce the search space, applying a greedy search strategy for efficient and coherent text reconstruction. Moreover, we assess the degree of information leakage in FedLLMs, focusing on how much input data can potentially be inferred from a single gradient update. This provided insights into balancing privacy protection and performance trade-offs.

To validate the effectiveness of our attacks, we perform experiments using GPT-2 \cite{gpt}, BERT \cite{bert}, and Qwen2.5-7B \cite{qwen2.5} on the popular datasets CoLA \cite{cola}, SST \cite{sst}, and Rotten Tomatoes \cite{rt} for batch sizes up to 128. We demonstrate that our attack can reconstruct large batch sizes with ROUGE-1/2 > 99 \cite{rouge} in cases where the parameters are frozen. UTR significantly outperforms the state-of-the-art (SOTA) baselines (i.e., LAMP~\cite{LAMP} and DAGER~\cite{dager}), achieving improvements of at least 26 points for small batch sizes and 90 points for larger ones in terms of reconstruction quality (ROUGE-1/2). We also conduct defensive experiments by incorporating Differential Privacy (DP)~\cite{dwork2006differential} and Gradient Pruning (GP)~\cite{xue2024revisiting} into the adapter updates. The results show that DP is a more effective defense method compared to GP, but it significantly undermines the model effectiveness and is not applicable to the real world. We summarize our contributions here:
\begin{icompact}
	\item We demonstrate that even when most model parameters are frozen, adapter-based FedLLMs can still be reconstructed by an adversary with high probability.
	
	\item We propose a UTR attack to efficiently reconstruct specific discrete inputs in FedLLMs. Furthermore, we tailor the attack algorithm to different usage conditions.
	
	\item We show the extent of information leakage in FedLLMs, focusing on the conditions under which the model leaks more information based on the provided input. We find that the unidirectional attention mechanism, commonly used in generative LLMs, may pose a significant risk of information leakage.
	
    \item We experimentally show that our algorithm can not only reconstruct the input exactly but also effectively reconstruct complete batches of size $b$ up to $128$ in the adapter-based FedLLMs. The results indicate that our algorithm can also achieve high accuracy with ROUGE-1/2 > 99.
\end{icompact}

\begin{figure}[!t]
	\centering
	\includegraphics[width=1.0\columnwidth]{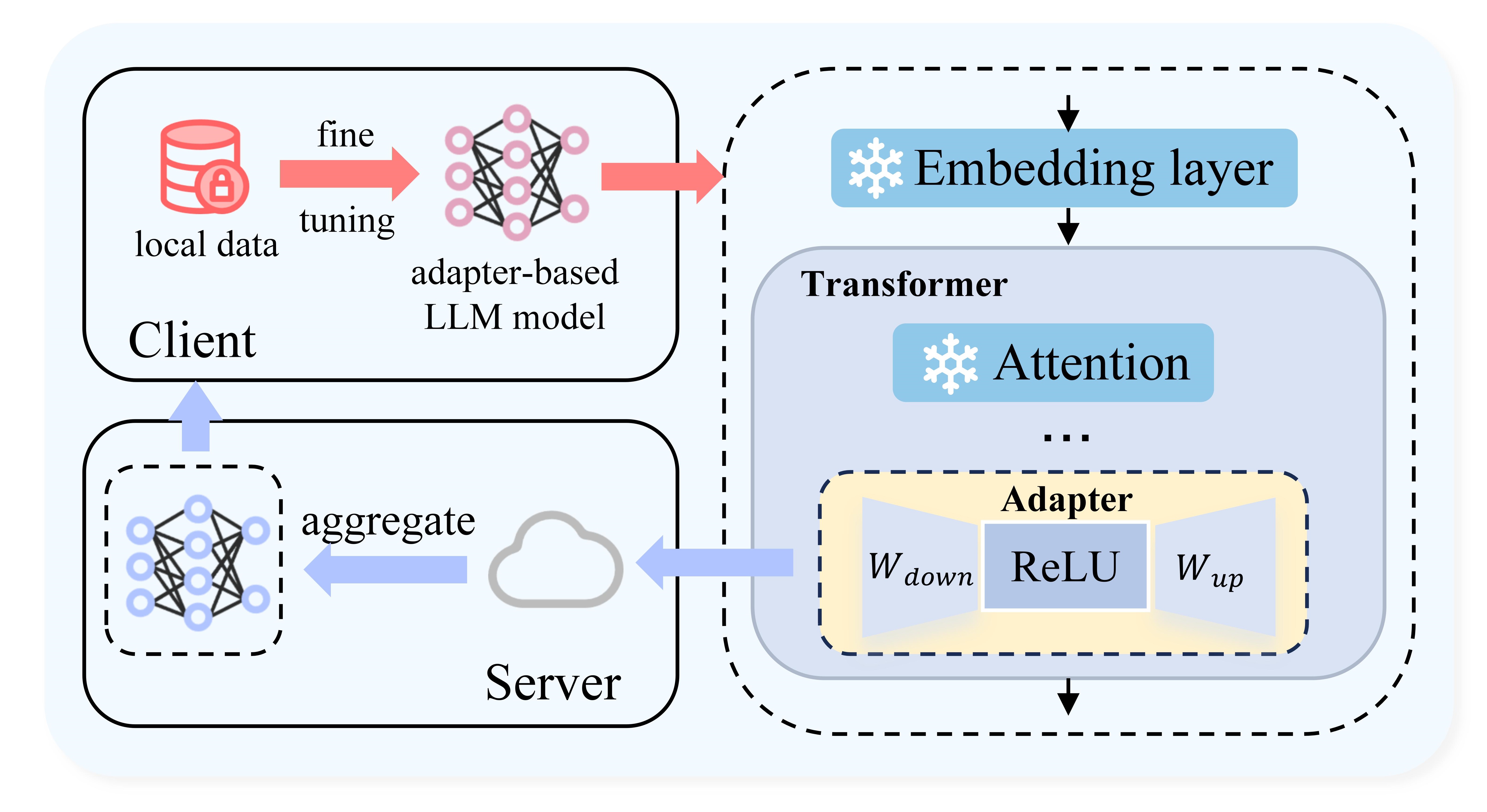}
	\caption{Standard architectures of adapter-based FedLLMs.}\label{arch}
	\vspace{-0.5cm}
\end{figure}

\section{Related Work}
\label{B}

\subsection{Adapter-based FedLLMs}
\label{B1}

With the continuous development of LLMs, the parameter size of these models has increased significantly \cite{cai2023efficient}, amplifying the need for high-quality public data. FedLLMs naturally emerge as a solution to this challenge and have become a key direction in LLM development. However, due to the enormous parameter size of LLMs, directly applying FL for pre-training or full-parameter fine-tuning leads to substantial computational, storage, and resource overhead, which is impractical for most LLM users~\cite{Adapter_efficient1}. Therefore, reducing resource overhead is the primary challenge for achieving FedLLMs.

Current approaches include partial parameter fine-tuning \cite{Bitfit}, adapter-based fine-tuning \cite{FedLLMs_FedDAT,adapter_Heterogeneous,MAD-X}, LoRA fine-tuning \cite{LoRA,hu2022lora,kuang2024federatedscope}, prompt tuning \cite{Parameter-Efficient_Prompt_Tuning,guo2023pfedprompt}, and prefix tuning \cite{P-tuningV2}. Among these, adapter-based fine-tuning, as illustrated in Figure~\ref{arch}, provides a practical trade-off between efficiency, performance, and privacy. By freezing the base model and training only lightweight adapter modules, this approach substantially reduces both communication and computation costs, making it well-suited for resource-constrained clients. It preserves the pretrained knowledge of the backbone while enabling effective personalization and multi-task adaptation across clients. Because only small adapter updates are shared, privacy risks and aggregation overhead are mitigated, while the modular design supports seamless integration of task- or domain-specific expertise \cite{long2024dual}. Overall, adapter-based FedLLMs achieve better scalability and practicality than full fine-tuning \cite{Adapter_efficient1}, while typically offering stronger performance and flexibility than simpler approaches such as prompt or prefix tuning.

\subsection{Gradient Inversion Attacks against FedLLMs}
Gradient inversion attacks (GIAs) \cite{Gradient_Inversion_Attack} aim to reconstruct training data by exploiting shared model gradients or parameters. These attacks pose severe privacy risks by exposing sensitive client information, motivating extensive research into privacy-preserving learning techniques. Early GIAs primarily targeted the image domain, where methods such as DGL \cite{zhu2019deep} reconstruct images by minimizing the distance between reconstructed and original gradients. Building upon DGL, TAG \cite{TAG} adapts this idea to Transformer-based models, improving textual reconstruction quality and reducing divergence from the true data. To overcome the inherent difficulty of recovering discrete text from continuous gradients, LAMP \cite{LAMP} introduces an alternating continuous–discrete optimization strategy, achieving better reconstruction fidelity than TAG.

A conceptual shift was introduced by FILM \cite{recovering}, which departs from optimization-based approaches. FILM leverages gradient information from the embedding layer and the discrete nature of textual inputs to infer words directly. However, its performance is constrained by limited token sets and heuristic inference rules. DAGER \cite{dager} improves FILM by exploiting the one-to-one mapping between embedding outputs and discrete tokens, achieving nearly perfect reconstruction for small vocabularies—but at the cost of substantial computational inefficiency due to brute-force search. 

In summary, existing GIAs largely focus on FedLLMs employing full-parameter fine-tuning, relying on gradients from embedding or Transformer layers for reconstruction. However, full-parameter tuning is computationally expensive and unrepresentative of mainstream adapter-based FedLLM paradigms, which typically freeze pre-trained parameters to reduce resource consumption. Consequently, the privacy risks of adapter-based FedLLMs remain underexplored, highlighting the need for systematic evaluation of their vulnerability to reconstruction attacks.

\section{Background and Threat Model}
\label{s3}


\subsection{Adapter-based FedLLMs}
\para{Overview.} We begin by introducing the core components of the adapter-based FedLLM architecture, as illustrated in Fig. \ref{arch}. The system comprises a server, multiple clients, the frozen parameters of a pre-trained LLM, and trainable adapter modules that can be inserted either within or between layers. The primary purpose of these adapters is to substantially reduce the number of trainable parameters, thereby lowering the computational, storage, and communication overhead. Based on their insertion locations, adapters can be categorized into two types: embedding adapters and layer adapters. Adapters placed after the embedding layer are referred to as embedding adapters, while those inserted within Transformer layers are called layer adapters. For example, the MAD-X~\cite{MAD-X} framework integrates adapters after the embedding layer, achieving strong cross-lingual and cross-task performance~\cite{MAD-X, MAD-X_1, MAD-X_2}, and has since been incorporated into the Hugging Face library~\cite{hugging_face_adapter}.

\para{Adapter.} As depicted in Fig. \ref{arch}, each adapter is a lightweight, trainable module typically composed of three components: a downsampling layer, an activation function, and an upsampling layer~\cite{adapter_structure,AdapterFusion}. The downsampling layer projects high-dimensional hidden representations into a lower-dimensional bottleneck space, while the upsampling layer restores them to the original dimensionality. Although specific adapter implementations may vary, most share this fundamental design. A key hyperparameter in adapter configuration is the bottleneck dimension ($\mathit{d_{bottleneck}}$.), which is determined by the reduction factor ($\mathit{reduction_{factor}}$). The reduction factor defines the ratio between the hidden layer dimension ($d_{hidden}$) and the bottleneck dimension, formally expressed as:
\begin{equation}
\mathit{reduction_{factor} = \frac{d_{hidden}}{d_{bottleneck}}}
\end{equation}

\para{Discussion.} This paper focuses exclusively on adapters following the aforementioned architecture, as they are the most widely adopted in real-world deployments. This
architecture is the basis for influential frameworks like MAD-X and is a standard option in
libraries like FATE and Hugging Face \textit{adapters}. On one hand, most existing studies employ this design, ensuring consistency and comparability with prior work. On the other hand, this architecture aligns with practical considerations in FedLLM scenarios, where computing resources are typically constrained, making parameter-efficient designs essential for feasible and scalable deployment. The attack would not directly apply to other PEFT approaches, which are out of scope.

\subsection{Threat Model}

\para{Adversary's Goal.} Following the attack goals established in prior work, the adversary seeks to reconstruct as much of a victim client’s training data as possible using only the observable gradient information, thereby compromising client privacy. We consider the attack successful when the reconstructed output meaningfully matches the original training data—for example, when it attains high semantic similarity or otherwise preserves the sensitive content of the true examples.

\para{Adversary's Identity.} The adversary could be any entity aiming to extract private information from the client. Examples include, but are not limited to, an honest-but-curious server, a rogue employee of a server company, a malicious third party controlling the server, a competing company, or a hacker. The method of obtaining the client's gradients and the associated difficulty vary among these adversaries. For instance, an honest-but-curious server can efficiently perform the attack due to its role in model selection, parameter distribution, and gradient aggregation, which inherently provides the necessary knowledge and access. Similarly, rogue employees benefit from comparable advantages. In contrast, adversaries such as competing companies or hackers may need to steal gradient information actively through targeted attacks. For simplicity, we focus on the honest-but-curious server as the adversary in this study, representing the most common attack model in FedLLMs.

\para{Adversary's Background Knowledge.} As noted above, we assume that the adversary has full knowledge of the pre-trained model’s architecture and parameters. In addition, since FedLLM training typically proceeds over multiple epochs, the adversary is assumed to observe the gradients shared by the victim client during each training round. However, the adversary cannot collude with any subset of clients and has no direct access to, or prior knowledge of, the client’s local training data.

\section{Our Proposed UTR}
\subsection{Overview}
To address the unique challenges of FedLLMs, i.e., the low dimensionality of updates, the frozen model backbone, and the extensive solution space, we propose UTR attack (see Fig.~\ref{framework}), a two-stage pipeline designed to systematically tackle these core issues. The key insight behind UTR is that coarse token recognition can be progressively refined into fine-grained sequence reconstruction by exploiting the intrinsic structural properties of adapter modules and frozen LLMs, enabling high-fidelity recovery of private training data. Next, we introduce the workflow of each stage as follows:

\para{Stage 1: Word Bag Inference.} This initial stage targets the embedding adapter to infer an unordered set of vocabulary tokens present in the client's private input. By leveraging the frozen embedding layer and the gradients from the embedding adapter, this stage bypasses traditional continuous optimization. It directly reduces the reconstruction search space, thereby addressing the challenge of low-dimensional updates (i.e., C1) (see \S\ref{Sub:WI}).

\para{Stage 2: Data Inference.} This subsequent stage takes the output ``word bag'' from the first stage and assembles coherent text sequences. It employs a constrained search, guided by grammatical, semantic, and contextual filters to prune invalid token combinations. By using gradients from the layer adapter to perform membership inference, this stage navigates the extensive solution space (i.e., C3) and compensates for the inaccessibility of the frozen backbone (i.e., C2) (see \S\ref{Sub:DI}).

\para{Discussion.} UTR deploys the two stages in a sequential pipeline, enabling near-exact reconstruction of client data even at large batch sizes and revealing a critical vulnerability. Importantly, UTR is model-agnostic, targeting the adapter modules as a universal attack surface to achieve high effectiveness across various LLM architectures and adapter configurations. Finally, we establish a Text Reconstruction Ability Assessment (see \S\ref{TRAA}) to quantify the fundamental limits of this attack. This assessment defines the attack's capability as the maximum number of words it can accurately reconstruct, which is fundamentally determined by the representational capacity of the low-rank space $S$ derived from the adapter's gradients.

\begin{figure}[!t]
	\centering
	\includegraphics[width=\columnwidth]{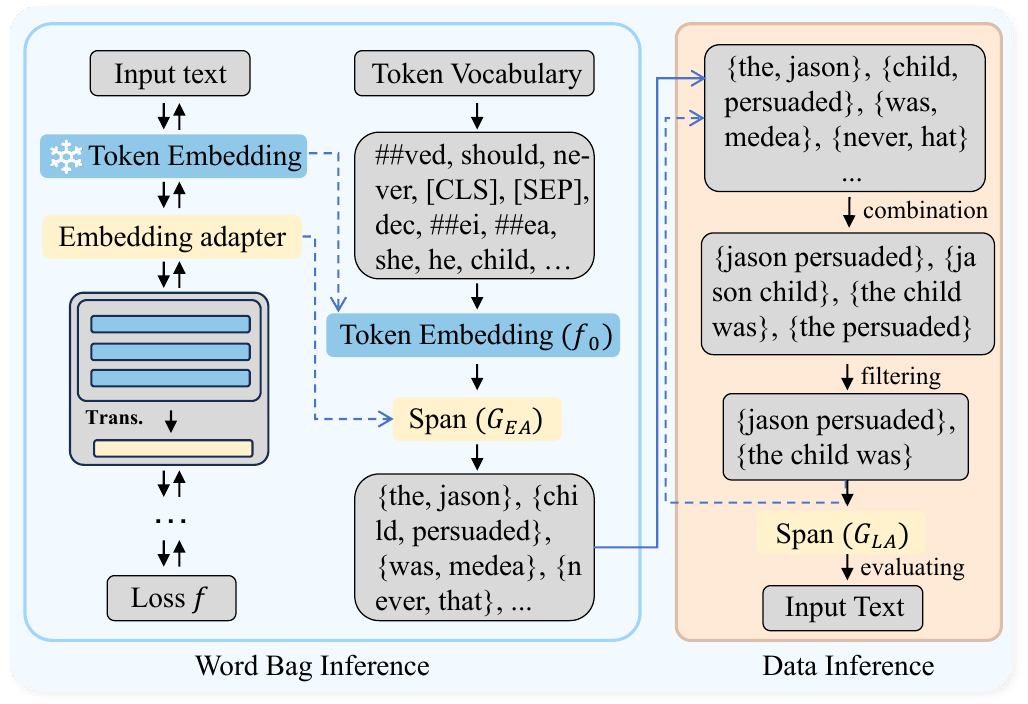}
	\caption{Overview of the proposed UTR attack.}\label{framework}
	\vspace{-0.5cm}
\end{figure}

\subsection{``Word Bag'' Inference Stage} 
\label{Sub:WI}

\para{Key Idea.} The first stage identifies vocabulary tokens likely present in the client's input. Exploiting gradients from the embedding adapter, UTR evaluates vocabulary items within the embedding space to determine token membership. This produces a ``word bag''—a set of candidate tokens inferred to appear in the original input. Compared with traditional optimization-based methods, this approach avoids continuous vector fitting and leverages the discrete mapping between embeddings and tokens, substantially reducing the reconstruction search space. Specifically, in adapter-based FedLLMs, frozen embeddings can be treated as a deterministic function $f_0$, which maps input word indices to embedding vectors. These embedding vectors are then passed to the embedding adapter as input. ``Word bag'' inference targets the embedding adapter by iterating over all words $x$ in the model's vocabulary, inputting them into the embedding layer, and checking whether the output $f_0(x)$ matches the embedding adapter's input data.

In this context, the ``word bag'' inference targeting an adapter essentially involves using the gradient information from the adapter to infer whether specific input data \(x\) belongs to the client's input dataset. Here, we first analyze the properties of neurons in a generally fully connected (FC) layer. Let the input to the fully connected layer be a vector \(\textbf{\textit{f}}(x)\). For any neuron in this layer with weight \(\mathbf{W}\) and bias \(b\), if the input vector \(\textbf{\textit{f}}(x)\) activates the neuron, then the output \(y\) of the neuron can be expressed as:
\begin{equation}
	\begin{split}
		y &= \mathrm{ReLU}(\mathbf{W} \cdot \textbf{\textit{f}}(x) +b) \\
		&= \begin{cases}
			\mathbf{W} \cdot \textbf{\textit{f}}(x) +b, & \text{if } \mathbf{W} \cdot \textbf{\textit{f}}(x) +b > 0 \\
			0, & \text{if } \mathbf{W} \cdot \textbf{\textit{f}}(x) +b \leq 0 
		\end{cases} \label{eq8}
	\end{split}.
\end{equation}
According to the chain rule, the gradients of the loss \(L\) with respect to the weight \(\mathbf{W}\) and bias \(b\) are calculated as follows:
\begin{equation}
	\begin{cases}
		\frac{\partial L}{\partial b} = \frac{\partial L}{\partial y} \cdot \frac{\partial y}{\partial b} \\
		\frac{\partial f}{\partial \mathbf{W}} = \frac{\partial L}{\partial y} \cdot \frac{\partial y}{\partial \mathbf{W}}=\frac{\partial L}{\partial y} \cdot \textbf{\textit{f}}(x)
	\end{cases}.\label{eq9}
\end{equation}

Since \({\partial y}/{\partial b}=1 \), substituting this into Eq. \eqref{eq9} yields:
\begin{equation}
	\frac{\partial f}{\partial \mathbf{W}} = \frac{\partial L}{\partial b} \cdot \textbf{\textit{f}}(x). \label{eq10}
\end{equation}

From Eq. \eqref{eq10}, \(\textbf{\textit{f}}(x)\) can be calculated as follows:
\begin{equation}
	\textbf{\textit{f}}(x) = \frac{\frac{\partial f}{\partial \mathbf{W}}}{\frac{\partial L}{\partial b}}= \frac{\nabla \mathbf{W}}{\nabla b}, \label{eq11}
\end{equation}
where \(\nabla \mathbf{W}\) and \(\nabla b\) represent the gradients of the weights and biases, respectively. The ratio \(\nabla \mathbf{W}/\nabla b\) is defined as the Ratio of Weight to Bias Gradient (RWBG).

When multiple vectors are fed into the fully connected layer in a batch \(\textbf{\textit{f}}(X)\), it is possible that more than one vector simultaneously activates the neuron. In this case, the relationship between the gradient of the neuron's weight and bias and the input data can be described as follows:
\begin{equation}
	\frac{\nabla \mathbf{W}}{\nabla b}= \sum_i \textbf{\textit{f}}(X)^i,  \label{eq12}
\end{equation}
where \(\textbf{\textit{f}}(X)^i\) is a vector in \(\textbf{\textit{f}}(X)\) which can activate the neuron.

For the network architecture shown in Fig.~\ref{arch}, by calculating the RWBGs for all neurons in FC1, a linear space \(\textbf{\textit{S}}\) can be constructed based on the RWBGs, as follows:
\begin{equation}
	\textbf{\textit{S}} =   \begin{bmatrix}
		\frac{\nabla \mathbf{W}_{m \cdot n}^{\mathrm{up}_1}}{\nabla  \mathbf{B}_{\mathrm{m}}^{\mathrm{up}_1}} \\
		\ldots \\
		\frac{\nabla \mathbf{W}_{m \cdot n}^{\mathrm{up}_i}}{\nabla  \mathbf{B}_{\mathrm{m}}^{\mathrm{up}_i}}
	\end{bmatrix}. \label{eq13}
\end{equation}

This linear space \(\textbf{\textit{S}}\) is essentially spanned by all the hidden embedding vectors fed into the linear layer. Since \(n>m \), the space \(\textbf{\textit{S}}\) is a low-rank space, which results in a limited representational capacity. The hidden embedding vectors of LLMs are strongly correlated with the input data. Therefore, if there exists a hidden embedding vector in  \(\textbf{\textit{f}}(X)\) belongs to the space \(\textbf{\textit{S}}\),  it can be inferred that \(X\) is part of the current training batch of the client. The method to determine whether a hidden embedding vector belongs to the space \(\textbf{\textit{S}}\) is to check if the hidden embedding vector can be expressed as a linear combination of the basis vectors of space \(\textbf{\textit{S}}\). This is a straightforward computation in linear algebra.

\subsection{Data Inference Stage}
\label{Sub:DI}

\para{Key Idea.} Following the Word Bag Inference stage, which produces an unordered set of candidate tokens $W_b$, the Data Inference stage aims to assemble these tokens into the original, coherent training text sequences. This stage systematically constructs and validates candidate sentences against the client's private data by leveraging gradients from the layer adapter and applying a series of filtering constraints.

The core of this stage is a structured search process, as formalized in Algorithm \ref{alg:atk} (Lines 12-34). The algorithm constructs candidate sentences by iteratively combining tokens from $W_b$. To manage the exponential complexity of this search space (C3), it employs a set of Boolean filter functions after each token extension: \textsf{EICW()} (Consecutive Identical Word Detection), \textsf{CG()} (Grammar Check) and \textsf{CS()} (Semantic Check). The specific implementation details are provided in Appendix \ref{algo}.

The final and critical step is membership inference, which verifies whether a candidate sentence matches the client's true training data. This is achieved by feeding the candidate sequence through the FedLLM to obtain the corresponding hidden embeddings $f_i(x)$ that serve as input to the first layer adapter. The core insight is that for the correct candidate, these hidden embeddings $f_i(x)$ must lie within the low-rank subspace $S_{LA}$ spanned by the basis vectors constructed from the client's shared layer adapter gradients (as derived in Eq. \eqref{eq12}). The verification is performed by checking if $f_i(x)$ can be expressed as a linear combination of the basis of $S_{LA}$, a direct computation in linear algebra.

\para{Impact of LLM Architecture.} The architecture of the underlying LLM influences the reconstruction outcome. For models with bidirectional attention mechanisms (e.g., BERT), the hidden embeddings are informed by the entire input context, enabling the reconstruction of the complete text in a single step. In contrast, for models with unidirectional attention (e.g., GPT), the hidden embedding for a token is only influenced by preceding tokens. This often results in the reconstruction of the text in a truncated or prefix-wise manner, requiring more candidate samples to assemble the full sequence.

Through this process of constrained candidate generation and gradient-based membership verification, the Data Inference stage successfully transforms an unordered bag of words into accurate reconstructions of the client's private training texts. The algorithm description of the technical details of the above two stages is summarized in Algorithm \ref{alg:atk} in the Appendix~\ref{algo}.

\subsection{Text Reconstruction Ability Assessment}
\label{TRAA}
The core mechanism of text reconstruction attacks involves inferring possible input words from the model's vocabulary (``word bag'') and combining them into coherent text. The attack's capability is therefore defined as the maximum number of words the algorithm can accurately reconstruct. When the client's input exceeds this threshold, accurate reconstruction becomes infeasible. 

We then analyze this threshold mathematically. Let a client's training batch contain $b$ sequences comprising a total of $n$ unique tokens. Each token $t_i$ is mapped through the frozen embedding layer to a hidden representation vector $v_i \in \mathbb{R}^{d_{\text{hidden}}}$. The downsampling layer of the adapter, implemented as a fully connected layer, performs the following transformation, i.e., $Z = \text{ReLU}(W V + b)$, where $W \in \mathbb{R}^{d_{\text{bottleneck}} \times d_{\text{hidden}}}$ is the weight matrix, $V \in \mathbb{R}^{d_{\text{hidden}} \times n}$ is the matrix of input hidden embeddings (each column corresponds to a token embedding $v_i$), $b \in \mathbb{R}^{d_{\text{bottleneck}}}$ is the bias vector, and $Z \in \mathbb{R}^{d_{\text{bottleneck}} \times n}$ is the activated output.

\para{Prior knowledge.} As mentioned in \S\ref{Sub:WI}, for the ``word bag'' inference to be successful, the hidden embedding vector $v$ of a token must lie within the low-rank subspace $S$ spanned by the client's actual input data. This is determined by analyzing the gradients. Following the Eqs. \eqref{eq10}--\eqref{eq13}, for a single neuron $j$ in the downsampling layer, the relationship between its weight gradient ($\nabla W_j$), bias gradient ($\nabla b_j$), and the input data is given by the RWBG:
\begin{equation}
    RWBG_j = \nabla W_j / \nabla b_j = \Sigma_{l \in A_j} v_l,   
\end{equation}
where $A_j$ is the set of indices of input vectors that activated neuron $j$ (i.e., $W_j \cdot v_l + b_j > 0$). The subspace $S$ is then constructed as the span of all these RWBG vectors across the $d_{bottleneck}$ neurons:
\begin{equation}
    S = span\{RWBG_1, RWBG_2, \ldots, RWBG_{d_{bottleneck}}\}    
\end{equation}
Since each $RWBG_j$ is a linear combination of the input vectors ${v_l}$, the subspace $S$ is a subspace of $span\{V\}$, the column space of the input embedding matrix. Thus, the representational capacity of the attack is fundamentally constrained by the rank of the subspace $S$.

\begin{theorem}\label{th-span-dim}
Let $V$ be a vector space over a field $F$. If $S$ is a subspace of $V$ spanned by a set of $m$ vectors, then $\dim(S) \leq m$.
\end{theorem}

\begin{theorem}\label{th-1}
(Rank-Capacity.) The maximum number of linearly independent hidden embedding vectors $v_i$ that can be perfectly represented in the subspace $S$ is at most $rank(S)$. Further imore, $rank(S) \le min(d_{bottleneck}, n, d_{hidden})$. 
\end{theorem}

\begin{proof}
 For the specific proof details of Theorems \ref{th-span-dim} and \ref{th-1}, please refer to the Appendix~\ref{proof}. Based on the theorem, we can get the Corollary below directly.   
\end{proof}

\begin{corollary}
(Maximum Recoverable Tokens.) The attack can, at most, uniquely identify and recover a number of tokens $k_{max}$ such that $k_{max} \le rank(S) \le d_{bottleneck}$.  
\end{corollary}

\begin{table*}[!t]
	\renewcommand{\arraystretch}{1.1}
	\caption{Attack performance evaluation results. The R-1 and R-2 denote the ROUGE-1 and ROUGE-2 scores, respectively.}
	\label{tab:performance}
	\centering
	\vspace{-5pt}
	\resizebox{0.65\textwidth}{!}{
		\begin{tabular}{cccccccccccc}
			\hline
			\multirow{2}{*}{\textbf{Batch}} & \multirow{2}{*}{\textbf{Method}} & \multirow{2}{*}{\textbf{Metrics}} & \multicolumn{3}{c}{\textbf{BERT}} & \multicolumn{3}{c}{\textbf{GPT2}} & \multicolumn{3}{c}{\textbf{Qwen}} \\ \cmidrule(r){4-6} \cmidrule(r){7-9} \cmidrule(r){10-12} 
			&  &  & \textbf{CoLA} & \textbf{SST} & \textbf{RT} & \textbf{CoLA} & \textbf{SST} & \textbf{RT} & \textbf{CoLA} & \textbf{SST} & \textbf{RT} \\ \hline
			\multirow{6}{*}{$B=1$} & \multirow{2}{*}{Ours} & R-1 & \textbf{100} & \textbf{100} & \textbf{100} & \textbf{100} & \textbf{100} & \textbf{100} & \textbf{100} & \textbf{100} & \textbf{100} \\ 
			&  & R-2 & \textbf{100} & \textbf{100} & \textbf{100} & \textbf{100} & \textbf{99} & \textbf{100} & \textbf{100} & \textbf{100} & \textbf{100} \\ \cline{2-12} 
			& \multirow{2}{*}{LAMP} & R-1 & 73.3 & 62.2 & 31.4 & 89.6 & 88.8 & 64.7 & / & / & / \\ 
			&  & R-2 & 43.3 & 31.8 & 9.3 & 51.9 & 56.8 & 16.5 & / & / & / \\ \cline{2-12} 
			& \multirow{2}{*}{DAGER} & R-1 & 100 & 100 & 100 & 100 & 100 & 100 & 0 & 0 & 0 \\ 
			&  & R-2 & 100 & 100 & 100 & 100 & 100 & 100 & 0 & 0 & 0 \\
			\hline
			\multirow{6}{*}{$B=2$} & \multirow{2}{*}{Ours} & R-1 & \textbf{100} & \textbf{100} & \textbf{100} & \textbf{99.83} & \textbf{100} & \textbf{100} & \textbf{100} & \textbf{100} & \textbf{100} \\ 
			&  & R-2 & \textbf{100} & \textbf{100} & \textbf{100} & \textbf{99.83} & \textbf{100} & \textbf{100} & \textbf{100} & \textbf{100} & \textbf{100} \\ \cline{2-12} 
			& \multirow{2}{*}{LAMP} & R-1 & 26.8 & 21.4 & 11.2 & 77.8 & 82.4 & 46.4 & / & / & / \\ 
			&  & R-2 & 11.0 & 9.2 & 0.9 & 31.5 & 45.7 & 7.6 & / & / & / \\ \cline{2-12} 
			& \multirow{2}{*}{DAGER} & R-1 & 100 & 99.3 & 98.1 & 100 & 100 & 100 & 0 & 0 & 0 \\ 
			&  & R-2 & 100 & 99.0 & 96.5 & 100 & 100 & 100 & 0 & 0 & 0 \\
			\hline 
			\multirow{6}{*}{$B=4$} & \multirow{2}{*}{Ours} & R-1 & \textbf{100} & \textbf{100} & \textbf{100} & \textbf{99.92} & \textbf{100} & \textbf{100} & \textbf{100} & \textbf{100} & \textbf{100} \\ 
			&  & R-2 & \textbf{100} & \textbf{100} & \textbf{100} & \textbf{99.88} & \textbf{99.25} & \textbf{100} & \textbf{100} & \textbf{100} & \textbf{100}\\ \cline{2-12} 
			& \multirow{2}{*}{LAMP} & R-1 & 13.4 & 9.8 & 6.3 & 66.2 & 69.5 & 35.1 & / & / & / \\ 
			&  & R-2 & 3.9 & 2.7 & 0.9 & 21.8 & 32.5 & 4.2 & / & / & / \\ \cline{2-12} 
			& \multirow{2}{*}{DAGER} & R-1 & 94.0 & 95.6 & 66.8 & 100 & 57.45 & 29.56 & 0 & 0 & 0 \\ 
			&  & R-2 & 89.9 & 93.0 & 50.1 & 100 & 52.63 & 21.65 & 0 & 0 & 0 \\
			\hline
			\multirow{6}{*}{$B=8$} & \multirow{2}{*}{Ours} & R-1 & \textbf{100} &\textbf{100} & \textbf{100}& \textbf{99.92} & \textbf{100} & \textbf{100} & \textbf{100} & \textbf{100} & \textbf{100} \\ 
			&  & R-2 & \textbf{100} & \textbf{100} & \textbf{100} & \textbf{99.83} & \textbf{99.63} & \textbf{100} & \textbf{100} & \textbf{100} & \textbf{100} \\ \cline{2-12} 
			& \multirow{2}{*}{LAMP} & R-1 & 8.9 & 8.1 & 6.8 & 52.9 & 56.9 & 27.3 & / & / & / \\ 
			&  & R-2 & 1.9 & 0.7 & 0.3 & 13.1 & 19.1 & 2.0 & / & / & / \\ \cline{2-12} 
			& \multirow{2}{*}{DAGER} & R-1 & 67.8 & 74.1 & 37.1 & 88.18 & 18.51 & 20.37 & 0 & 0 & 0 \\ 
			&  & R-2 & 48.8 & 59.8 & 11.4 & 84.50 & 7.98 & 13.06 & 0 & 0 & 0 \\
			\hline
		\end{tabular}
	}
    \vspace{-0.4cm}
\end{table*}

\begin{table}[!t]
	\centering
	\caption{Main results on the GPT2-Large and Qwen-7B models with large batch sizes on various datasets.}\label{tab:performance2}
	\vspace{-5pt}
	\resizebox{0.45\textwidth}{!}{
		\begin{tabular}{lccccccccc}
			\toprule
			\multirow{2}{*}{} & \multirow{2}{*}{} & \multicolumn{2}{c}{$B=16$} & \multicolumn{2}{c}{$B=32$} & \multicolumn{2}{c}{$B=64$} & \multicolumn{2}{c}{$B=128$} \\
			\cmidrule(r){3-4} \cmidrule(r){5-6} \cmidrule(r){7-8} \cmidrule(r){9-10}
			&  & R-1 & R-2 & R-1 & R-2 & R-1 & R-2 & R-1 & R-2 \\
			\midrule
			\multirow{2}{*}{\textbf{CoLA}} & GPT2 & $\mathbf{100.0 }$ & $\mathbf{100.0 }$ & $\mathbf{100.0 }$ & $\mathbf{100.0 }$ & $\mathbf{100.0 }$ & $\mathbf{100.0 }$ & $\mathbf{100.0 }$ & $\mathbf{100.0 }$ \\
			& Qwen & $\mathbf{100.0 }$ & $\mathbf{100.0 }$ & $\mathbf{100.0 }$ & $\mathbf{100.0 }$ & $\mathbf{100.0 }$ & $\mathbf{100.0 }$ & $\mathbf{100.0 }$ & $\mathbf{100.0 }$ \\
			\cmidrule(r){1-10}
			\multirow{2}{*}{\textbf{SST-2}} & GPT2 & $\mathbf{100.0 }$ & $\mathbf{99.4 }$ & $\mathbf{100.0 }$ & $\mathbf{99.7 }$ & $\mathbf{0.0 }$ & $\mathbf{0.0 }$ & $\mathbf{0.0 }$ & $\mathbf{0.0 }$ \\
			& Qwen & $\mathbf{100.0 }$ & $\mathbf{100.0 }$ & $\mathbf{100.0 }$ & $\mathbf{100.0 }$ & $\mathbf{100.0 }$ & $\mathbf{100.0 }$ & $\mathbf{100.0 }$ & $\mathbf{100.0 }$ \\
			\cmidrule(r){1-10}
			\multirow{2}{*}{\textbf{RT}} & GPT2 & $\mathbf{100.0 }$ & $\mathbf{100.0 }$ & $\mathbf{100.0 }$ & $\mathbf{100.0 }$ & $\mathbf{0.0 }$ & $\mathbf{0.0 }$ & $\mathbf{0.0 }$ & $\mathbf{0.0 }$ \\
			& Qwen & $\mathbf{100.0 }$ & $\mathbf{100.0 }$ & $\mathbf{100.0 }$ & $\mathbf{100.0 }$ & $\mathbf{100.0 }$ & $\mathbf{100.0 }$ & $\mathbf{100.0 }$ & $\mathbf{100.0 }$ \\
			\bottomrule
	\end{tabular}}
	\vspace{-0.5cm}
\end{table}

\section{Experiments}

\subsection{Experimental Setup}
\para{Model and Dataset.} We will mainly carry out the main experiments on three datasets, i.e., CoLA~\cite{cola}, SST~\cite{sst}, and Rotten Tomatoes (RT)~\cite{rt}. We evaluate the attack performance of the proposed UTR attack on three models: Bert-Base \cite{bert}, GPT2-Large \cite{gpt}, and Qwen2.5-7B \cite{qwen2.5}. Detailed information about the models and datasets can be found in Appendix \ref{data}.  

\para{Training and Adapter Settings.} In the adapter-based FedLLMs, proper hyperparameter selection, i.e., $\mathit{reduction_{factor}}$, is a notable factor that determines the proposed attack performance. We investigate the performance of the proposed framework with different settings and try to find the best-performing threshold for it. In particular, we employ $\mathit{reduction_{factor}} \in \{1, 2, 4, 8\}$ to adjust the best threshold of the proposed framework. The structure of adapters is set as Fig.~\ref{arch}, and we set the hyperparameter $\mathit{reduction_{factor}}=2$ in the main experiments of our paper. The others will be presented in the \S\ref{sub:hy}. We insert the adapters between the model's layers rather than inside. During the experimentation, we utilized only the first two Transformer layers of the model. 

\para{Baselines.} We compare our attack with the following methods, including the optimization-based method (i.e., LAMP \cite{LAMP}) and the embedding-based method (i.e., DAGER \cite{dager}) designed for GIAs in FedLLMs. More details can be found in the Appendix~\ref{baselines}.

\para{Evaluation Metrics.} Following prior work, we evaluate the quality of text reconstruction using ROUGE-1/2 scores \cite{rouge}. ROUGE-1 and ROUGE-2 are widely adopted metrics for assessing the quality of text summarization or machine-generated text by comparing it to human-written reference summaries. ROUGE-N is computed as follows:
\begin{equation}
    ROUGE-N = \frac{\Sigma_i{Count_i(matched \ \ n-grams)}}{\Sigma_i{Count_i(n-grams \ \ in \ \ Reference_i)}}
\end{equation}

\subsection{Attack Performance Evaluation}
\label{I:utr}

In this section, we aim to evaluate the effectiveness of the proposed UTR attack against adapter-based FedLLMs. We compare UTR’s text reconstruction performance with baseline methods, i.e., LAMP and DAGER, across three benchmark datasets (CoLA, SST, and Rotten Tomatoes) and three representative model architectures (BERT, GPT2-large, and Qwen2.5-7B). An example of a UTR attack can be found in the Appendix~\ref{example}.

\para{Attack Performance under Different Data Volumes.} Specifically, we vary the batch size from 1 to 8 to evaluate performance under different data volumes. As shown in Table~\ref{tab:performance}, UTR achieves nearly perfect text reconstruction, with ROUGE-1/2 scores approaching 100\% across all models and datasets. In contrast, LAMP exhibits substantial performance degradation as batch size increases, while DAGER fails completely under the frozen backbone setting—highlighting UTR’s superior robustness and adaptability. This performance gap arises because LAMP formulates reconstruction as an optimization problem, often resulting in unstable convergence, whereas UTR directly exploits gradient information for deterministic reconstruction. DAGER’s failure, on the other hand, is attributed to its reliance on gradients from Transformer attention layers, which are typically frozen during fine-tuning. Note that as DAGER is not open-sourced, our implementation is based on its published description, which may lead to minor reproduction discrepancies.

\para{Attack Performance under Large Batch Size.} To evaluate scalability in practical large-batch scenarios, we evaluate the reconstruction capability of UTR on the GPT2-Large and Qwen2.5-7B models with substantially increased batch sizes. Specifically, we set the batch size to 16, 32, 64, and 128 on the CoLA, SST, and RT datasets. The results are shown in Table~\ref{tab:performance2}, where we discover that UTR maintains perfect or near-perfect reconstruction (ROUGE-1/2 $\approx$ 100) for Qwen and on CoLA for GPT2-Large, even at a batch size of 128. However, we find that the attack on GPT2-Large fails on the SST and RT datasets for batches larger than 32, indicating that the unidirectional attention mechanism in GPT models may pose a greater challenge for reconstructing longer sequences from large batches.

\subsection{Attack Performance Evaluation under Defenses}
This section empirically evaluates the defensive efficacy of Differential Privacy (DP) \cite{wei2020framework} and Gradient Pruning (GP) \cite{zhu2019deep}. We benchmark these defense methods against the UTR attack across model architectures, assessing their impact on both model utility and privacy protection. We implement a centralized DP framework where the server injects Gaussian noise into the aggregated adapter gradients from clients. The defense strength is primarily governed by the noise multiplier ($\sigma$), for which we explore a range of values $\sigma \in \{0.01, 1.5, 2.0, 3.0\}$ with a fixed clipping bound of $1.0$. As a heuristic defense, GP discards gradient elements with absolute values below a specified threshold, thereby perturbing the gradient signal. We evaluate pruning rates ($r$) $\in \{90\%, 95\%, 99.9\%\}$, where $r$ denotes the proportion of elements set to zero.

\begin{table}[!t]
	\centering
	\caption{Effectiveness of defenses against our attack.}\label{tab:performance_defence}
	\vspace{-6pt}
	\resizebox{0.45\textwidth}{!}{
		\begin{tabular}{cccccccc}
			\toprule
			\multirow{2}{*}{\textbf{Defense}} & \multirow{2}{*}{\textbf{Parameter}} & \multicolumn{2}{c}{\textbf{BERT}} & \multicolumn{2}{c}{\textbf{GPT2}} & \multicolumn{2}{c}{\textbf{Qwen}} \\
			\cmidrule(r){3-4} \cmidrule(r){5-6} \cmidrule(r){7-8} 
			&  & R-1 & R-2 & R-1 & R-2 & R-1 & R-2  \\
			\midrule
			\multirow{3}{*}{\textbf{DP}} 
			&$\sigma = 0.01 $ & $\sim 100\%$ & $\sim 100\%$ & $\sim 100\%$ & $\sim 100\%$ & $\sim 100\%$ & $\sim 100\%$  \\
			& $\sigma = 1.5$ & $\sim 15\%$ & $\sim 15\%$ & $\sim 13\%$ & $\sim 17\%$ & $\sim 0\%$ & $\sim 0\%$  \\
			& $\sigma = 2.0$ & $< 2\%$ & $< 2\%$ & $< 2\%$ & $< 2\%$ & $\sim 0\%$ & $\sim 0\%$  \\
			& $\sigma = 3.0$ & $\sim 0\%$ & $\sim 0\%$ & $\sim 0\%$ & $\sim 0\%$  & $\sim 0\%$ & $\sim 0\%$   \\
			\cmidrule(r){1-8}
			\multirow{3}{*}{\textbf{GP}} & $r = 90\%$ & $\sim 65\%$ & $\sim 65\%$ & $\sim 65\%$ & $\sim 65\%$ & $\sim 60\%$ & $\sim 60\%$  \\
			& $r = 99\%$ & $\sim 25\%$ & $\sim 25\%$ & $\sim 25\%$ & $\sim 25\%$ & $\sim 26\%$ & $\sim 25\%$ \\
			& $r = 99.9\%$ & $< 5\%$ & $< 5\%$ & $< 5\%$ & $< 5\%$ & $< 5\%$ & $< 5\%$ \\
			\bottomrule
	\end{tabular}}
	\vspace{-0.5cm}
\end{table}

\para{Attack Performance under DP.} As shown in Table~\ref{tab:performance_defence}, the UTR attack achieved a 100\% attack success rate at a noise level ($\sigma=0.01$) that does not destroy the effectiveness of the model. At $\sigma=1.0$, the defense became nearly impenetrable, with success rates dropping below 2\% and reconstructed text devolving into nonsensical word sequences (e.g., ``Aesthetic volume conference bicycle''). A setting of $\sigma=1.5$ rendered the attack completely ineffective, achieving a 0\% success rate. This demonstrates that DP provides a mathematically grounded and definitive defense against gradient inversion attacks, with its strength precisely controllable via the $\sigma$ parameter. \textbf{However, we demonstrate in the Appendix~\ref{A1} that such DP defense strength renders the model unusable.}

\begin{figure*}[!t]
	\centering
	\subcaptionbox{BERT \label{ass:a}}[0.32\textwidth]{\includegraphics[width=0.32\textwidth]{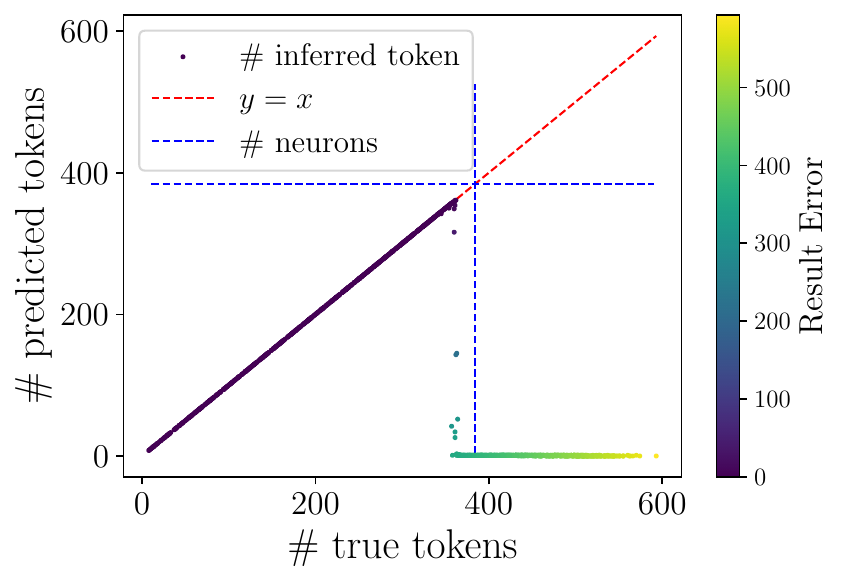}} %
	\hfill
	\subcaptionbox{GPT-Large \label{ass:b}}[0.33\textwidth]{\includegraphics[width=0.33\textwidth]{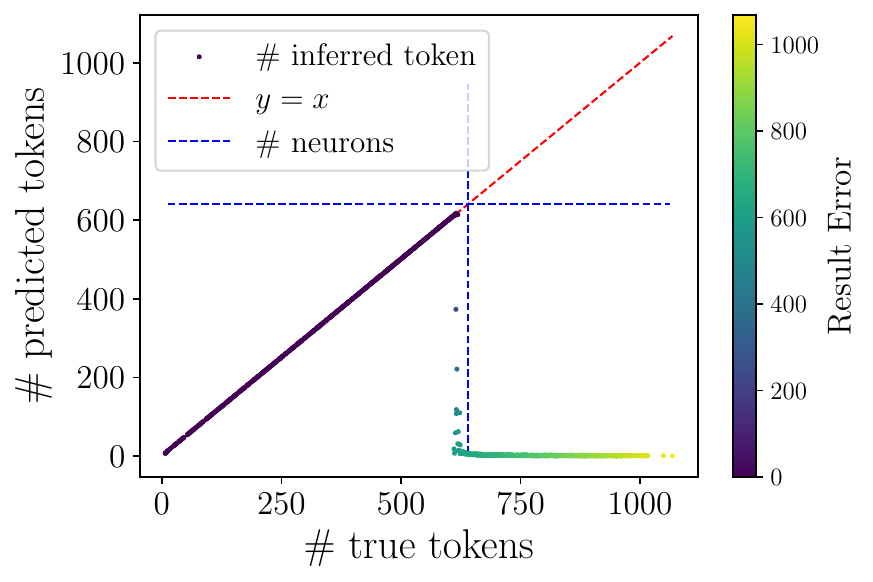}} %
	\hfill
	\subcaptionbox{Qwen-7B \label{ass:c}}[0.33\textwidth]{\includegraphics[width=0.33\textwidth]{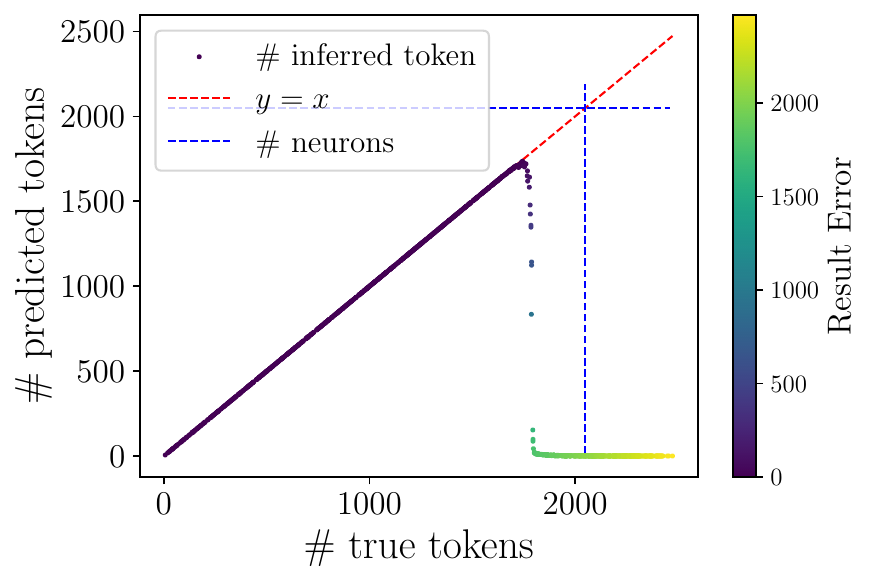}}
	\hfill
	\caption{Relationship between the number of recovered tokens in the gradient information of the embedding adapter and the number of true tokens.}\label{assess}
\end{figure*}

\begin{figure}[!t]
	\centering
	\subcaptionbox{$reduction_{factor}$     \label{fig:placeholder}}[0.22\textwidth]{\includegraphics[width=0.22\textwidth]{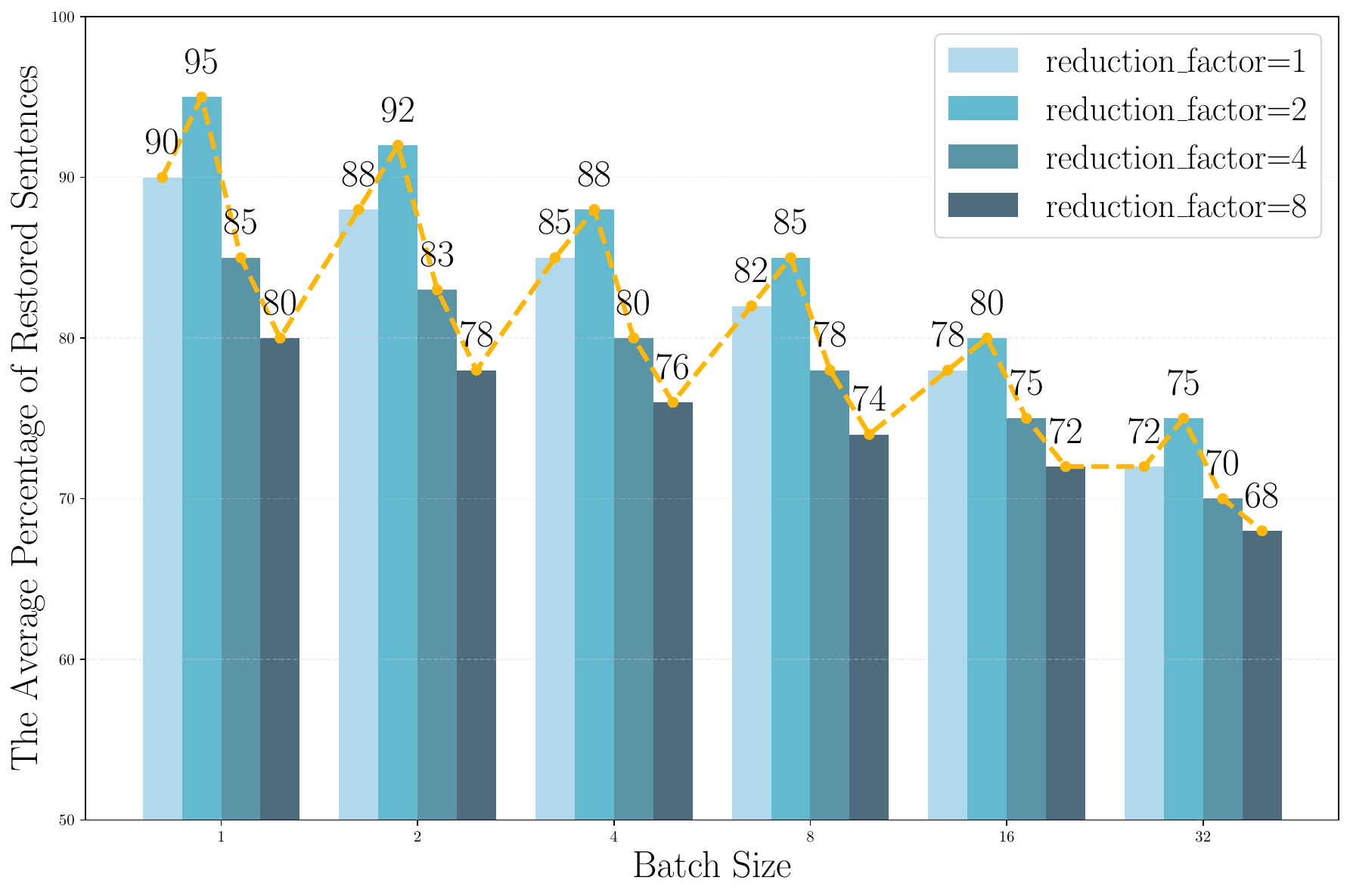}} %
	\subcaptionbox{$atol$  \label{fig:placeholder2}}[0.22\textwidth]{\includegraphics[width=0.22\textwidth]{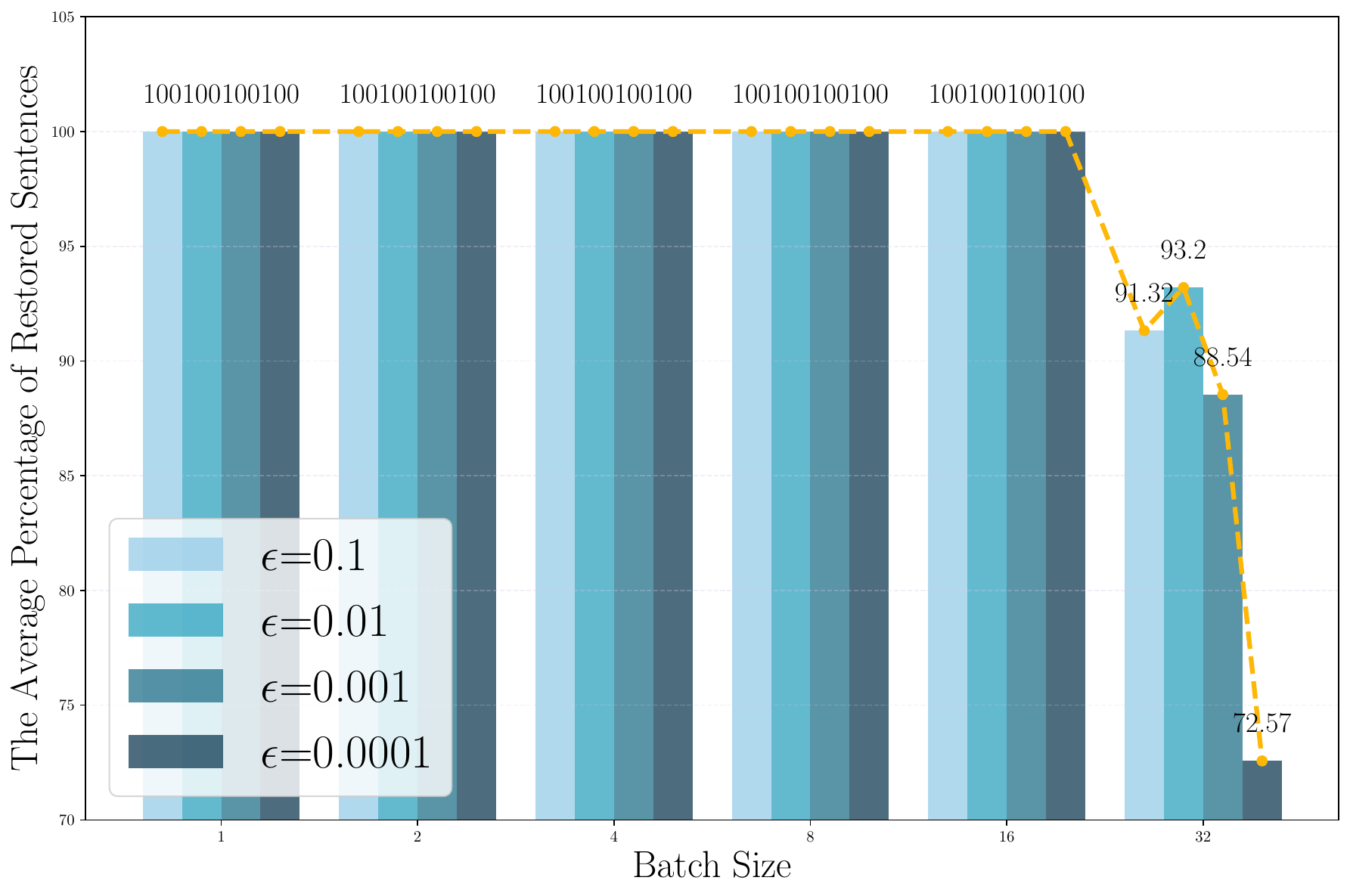}} %
	\caption{The performance across various $reduction_{factor}$ and $atol$ settings.}\label{asse}
	\vspace{-0.5cm}
\end{figure}

\para{Attack Performance under GP.} In contrast, GP offers fragile and unreliable protection. Under a mild pruning regime ($r = 90\%$), the attack remained highly viable with a success rate of ~65\%, as the retained top gradients contained sufficient information for reconstruction. Aggressive pruning ($r = 99\%$) reduced the success rate to approximately 25\%, yet this still permitted significant leakage of key tokens and semantic concepts. It was only at an extreme, utility-destroying threshold ($r = 99.9\%$) that the attack was effectively thwarted (success rate <5\%). This fundamental limitation underscores that GP's protection is heuristic and brittle, failing to securely safeguard data without simultaneously negating the model's learning capacity.

\subsection{Impact of the Adapter Settings}
\label{sub:hy}
In this section, we further investigate the impact of the adapter's bottleneck size and the span similarity threshold, controlled by the $reduction_{factor} \in \{1, 2, 4, 8\}$, and the $\epsilon \in \{0.1,0.01,0.001,0.0001\}$, on the success of our reconstruction attacks. 
 
\para{Impact of $reduction_{factor}$.}
The results, illustrated in Fig.~\ref{fig:placeholder}, indicate a strong inverse correlation between the $reduction_{factor}$ and attack performance. As the $reduction_{factor}$ increases, the dimensionality of the adapter's hidden representation decreases, creating a more severe information bottleneck. This bottleneck directly degrades the quality and completeness of the gradient information available for inversion. Consequently, while a larger $reduction_{factor}$ offers improved parameter efficiency for federated tuning, it also inherently enhances privacy by attenuating the signal that can be exploited by gradient-based attacks.

\para{Impact of Span Similarity Threshold $\epsilon$.}
The experimental results, illustrated in Fig.~\ref{fig:placeholder2} reveal a critical trade-off between the interaction of batch size and solver tolerance (Atol). The key finding is that a higher tolerance value allows the attack to maintain a perfect or near-perfect success rate across a wide range of batch sizes, demonstrating that efficient large-batch processing can be achieved without performance loss. Conversely, when a very strict tolerance is enforced, the attack becomes highly fragile; its performance remains perfect for small batches but collapses dramatically as the batch size increases. This indicates that striving for high numerical precision undermines the attack's stability when scaling computations.

\subsection{Effectiveness Analysis of UTR Attack}
\label{s7}
This section evaluates the fundamental capability of adapter gradients to capture meaningful information, which directly dictates the upper bound of text reconstruction, as described in \S\ref{TRAA}. To validate this, we conducted a large-scale experiment on the CoLA dataset, correlating the number of tokens reconstructed via the embedding adapter's gradients against the true number of tokens in the original text. We varied the batch size from 1 to 256 to ensure a wide distribution of token counts, performing five rounds for a total of 1,280 trials. 

The results, plotted in Fig.~\ref{assess}, show the relationship between the inferred and true token counts. Two key patterns emerge from the data. First, the effective information captured by the gradients is consistently slightly less than the total number of downsampling layer neurons. This aligns with our theoretical analysis in \S\ref{TRAA}. Second, the amount of information effectively captured is positively correlated with the size of the adapter's downsampling layer. A comparison of the subfigures in Fig.~\ref{assess} confirms that the adapter in Qwen, with the most neurons, captures the most information. This pattern directly explains the superior attack performance of Qwen noted in prior sections.

\section{Conclusion}

This paper investigates the privacy risks of adapter-based FedLLMs under gradient inversion attacks (GIAs). We propose a unified framework for analyzing GIAs and uncover that adapter structures create new privacy vulnerabilities. To exploit these, we introduce UTR for data reconstruction by leveraging discrete LLM properties (e.g., positional embeddings, self-attention). Finally, we quantify information leakage from single gradient updates, offering insights into the privacy–performance trade-off. Our work advances understanding and defense of privacy in adapter-based FedLLMs.

\begin{acks}
This work was supported by the National Natural Science Foundation of China (Grant No. 62472431) and the Innovation Research Foundation of the National University of Defense Technology (Grant No. ZK25-17).
\end{acks}


\bibliographystyle{ACM-Reference-Format}
\balance
\bibliography{sample-base}

\clearpage
\appendix

\section{Proofs}\label{proof}
The representational capacity of the attack is fundamentally constrained by the rank of the subspace $S$. We then prove Theorems \ref{th-span-dim} and \ref{th-1} as follows:

\begin{proof}
Let $W = {\mathbf{w}_1, \mathbf{w}_2, \dots, \mathbf{w}_m}$ be a spanning set for $S$. If $W$ is linearly independent, then $\dim(S) = m$. If $W$ is linearly dependent, then there exists a proper subset of $W$ that forms a basis for $S$, and thus, $\dim(S) < m$. Therefore, in all cases, $\dim(S) \leq m$.
\end{proof}


\begin{proof}
We prove the theorem by establishing the bounds on the rank of $S$ through the following arguments:

\para{Step 1:} By definition, the subspace $S$ is spanned by the $d_{\text{bottleneck}}$ vectors ${RWBG}_{j}$. By Theorem \ref{th-span-dim}, the dimension of a subspace cannot exceed the number of vectors in its spanning set, therefore we have:
\begin{equation}
rank(S)\le d_{bottleneck}.
\end{equation}
\para{Step 2:} Each $RWBG_j$ is a linear combination of the $n$ input hidden embedding vectors ${v_i}^n$. Therefore, $S$ is a subspace of the span of ${v_i}^n$, which has dimension at most $\min(n, d_{hidden})$. Thus:
\begin{equation}
rank(S) \le min(n,d_{hidden})
\end{equation}
\para{Step 3:} Combining the bounds from Step 1 and Step 2, the most restrictive bound is the minimum of these three values. In practical FedLLM scenarios, $d_{bottleneck}$ is typically much smaller than both n and $d_{hidden}$ (e.g., $d_{bottleneck}=512$, $d_{hidden}=4096$, $n$ can be in the thousands for a large batch). Therefore, the bottleneck dimension is the effective upper limit: 
\begin{equation}
rank(S) \le min(d_{bottleneck}, n, d_{hidden})
\end{equation}

\noindent Finally, the maximum number of linearly independent vectors that can be perfectly represented in $S$ is precisely $\text{rank}(S)$, by the definition of dimension. This completes the proof.
\end{proof}

We analyze UTR attacks within an ideal setting to delineate the theoretical limits of their efficacy. By characterizing this upper bound, we provide the necessary analytical groundwork to inform the design of provably resilient countermeasures.

\section{Algorithm}\label{algo}
The pseudo code of the detailed execution steps of the UTR attack is shown in the following algorithm. The algorithm constructs candidate sentences by iteratively combining tokens from $W_b$. To manage the exponential complexity of this search space (C3), it employs a set of Boolean filter functions after each token extension: \texttt{EICW()} (Consecutive Identical Word Detection): Prunes candidates containing repeated consecutive tokens, which are uncommon in valid text. \texttt{CG()} (Grammar Check): Filters out sequences that exhibit obvious syntactic errors. It is evaluated using \textit{language\_tool\_python} with over 50 syntactic rules. \texttt{CS()} (Semantic Check): Discards combinations that lack basic semantic coherence. We measure it via cosine similarity > 0.2 in SBERT embedding space. These filters collectively guide the search towards linguistically plausible sequences, drastically reducing the number of candidates that require computationally expensive verification.

\begin{algorithm}[!t]
	\caption{UTR Attack}
	\label{alg:atk}
	\begin{algorithmic}[1]
		\Require 
		\parbox[t]{0.85\linewidth}{
			Embedding-Adapter gradients $G_{EA}$, \\Layer-Adapter gradients $G_{LA}$, \\vocabulary $\mathcal{W}$,\\ projection thresholds $\varepsilon_{EA}$ and $\varepsilon_{LA}$, \\maximum reconstruction length $L_{\max}$,\\ and constrained beam width $B$
		}
		\Ensure Reconstructed sentence candidates $\mathcal{S}_{\text{cand}}$
		
		\Statex \textbf{Preprocessing:}
		\State Compute subspace bases: $S_{EA} \gets \text{span}(G_{EA})$, $S_{LA} \gets \text{span}(G_{LA})$
		\State Define $\mathrm{in\_span}(v, S, \varepsilon)$ as the predicate \\
		$\|v - P_S(v)\|_2 < \varepsilon$, where $P_S(v)$ is the orthogonal projection of $v$ onto $S$
		
		\Statex \textbf{Step 1: Infer Candidate Word Bag}
		\State Initialize candidate word bag: $\mathcal{W}_b \gets \emptyset$
		\ForAll{$w \in \mathcal{W}$}
		\State $v_w \gets f_0(w)$ \Comment{Word-level representation of $w$}
		\If{$\mathrm{in\_span}(v_w, S_{EA}, \varepsilon_{EA})$}
		\State $\mathcal{W}_b \gets \mathcal{W}_b \cup \{w\}$
		\EndIf
		\EndFor
		\If{$|\mathcal{W}_b| = 0$}
		\State \Return $\emptyset$ \Comment{No recoverable words found}
		\EndIf
		
		\Statex \textbf{Step 2: Sentence Reconstruction via Beam Search}
		\State Initialize beam: $\mathcal{B} \gets \{(\text{seq} = [\,],\; \text{score} = 0)\}$
		\For{$l = 1$ to $L_{\max}$}
		\State $\mathcal{B}_{\text{new}} \gets \emptyset$
		\ForAll{$(\text{seq},\, \text{score}) \in \mathcal{B}$}
		\ForAll{$w \in \mathcal{W}_b$}
		\If{\textsf{EICW}(\text{seq}, $w$)} \Comment{Exclude if $w$ already in \text{seq}}
		\State \textbf{continue}
		\EndIf
		\State $\text{seq}' \gets \text{seq} \oplus [w]$ \Comment{Append $w$ to sequence}
		\If{\textsf{CG\_score}(\text{seq}') \textbf{or} \textsf{CS\_score}(\text{seq}')} 
		\State \textbf{continue} \Comment{Skip if grammatically implausible or semantically incoherent}
		\EndIf
		\State $v_{\text{seq}'} \gets f_1(\text{seq}')$ \Comment{Sequence-level representation}
		\State $p_s \gets \mathrm{span\_similarity}(v_{\text{seq}'}, S_{LA})$
		\If{$\mathrm{in\_span}(v_{\text{seq}'}, S_{LA}, \varepsilon_{LA})$}
		\State Add $(\text{seq}',\, p_s)$ to $\mathcal{B}_{\text{new}}$
		\EndIf
		\EndFor
		\EndFor
		\If{$\mathcal{B}_{\text{new}} = \emptyset$}
		\State \textbf{break} \Comment{No valid extensions; terminate early}
		\EndIf
		\State \textbf{Prune:} Keep top-$B$ sequences in $\mathcal{B}_{\text{new}}$ with highest $p_s$
		\State $\mathcal{B} \gets \mathcal{B}_{\text{new}}$
		\If{convergence or early-stopping criterion is met}
		\State \textbf{break}
		\EndIf
		\EndFor
		
		\Statex \textbf{Step 3: Return Final Candidates}
		\State Sort $\mathcal{B}$ by descending $p_s$ to obtain $\mathcal{S}_{\text{cand}}$
		\State \Return $\mathcal{S}_{\text{cand}}$
	\end{algorithmic}
\end{algorithm}

\section{Models and Datasets}\label{data}
Table~\ref{tab:model} provides key information about these models, including whether their embedding layers incorporate positional encoding mechanisms, whether they utilize unidirectional or bidirectional attention mechanisms, and their hidden dimension sizes. These features are identified as critical because they influence the implementation of the attack and impact the effectiveness of the attack. A detailed analysis of these factors will be conducted during the experimental process.

\begin{table}[!t]
	\centering
	\caption{Basic information of models. PE means whether the model has position embedding.}
	\label{tab:model}
	\begin{tabularx}{0.45\textwidth}{cccc}
		\toprule
		\textbf{Model} & \textbf{PE} & \textbf{Self/Bidirect} & \textbf{Hidden Embedding}\\
		\midrule
		BERT-Base & \checkmark & Bidirect & 786\\
		\midrule
		GPT2-Large & \checkmark & Self & 1280\\
		\midrule
		Qwen2.5-7B & $\times$ & Bidirect & 4096\\
		\bottomrule
	\end{tabularx}
	\vspace{-0.3cm}
\end{table}

This work leverages three benchmark datasets, i.e., the Corpus of Linguistic Acceptability (CoLA), the Stanford Sentiment Treebank (SST-2), and the Rotten Tomatoes (RT) movie review dataset, each probing a distinct facet of this capability. The numerical information of the above dataset can be found in Table \ref{tab:dataset}.

\para{CoLA Dataset.} CoLA is a specialized dataset designed to evaluate a model's fundamental grasp of English syntax and grammatical structure. Unlike tasks focused on meaning or sentiment, CoLA frames language understanding as a binary classification problem of grammaticality judgment, where a model must determine whether a given sentence is linguistically acceptable or unacceptable. The sentences are drawn from expert linguistics publications and annotated by trained linguists, ensuring high-quality labels for complex grammatical phenomena. This includes testing for subtle violations in subject-verb agreement, verb argument structure, and anomalous word sequences. As a core component of the GLUE benchmark, CoLA provides a crucial test of a model's syntactic competence, probing its ability to internalize the formal rules of the language rather than just statistical patterns in the data.

\para{SST Dataset.} SST-2 offers a comprehensive benchmark for fine-grained sentiment analysis by moving beyond simple document-level classification. Derived from movie reviews, this dataset is unique in that it provides sentiment labels not only for entire sentences but for every single phrase and constituent within their parse trees. This rich, granular annotation allows models to be trained and evaluated on understanding how sentiment is compositionally built from smaller units, capturing the nuanced impact of negations, contrasts, and intensifiers. The task is a binary classification of sentiment into positive or negative categories. SST-2's primary challenge and significance lie in pushing models to develop a more profound, structured understanding of sentiment composition, making it a standard for assessing nuanced semantic understanding.

\para{RT Dataset.} RT movie review dataset serves as a classic and widely-adopted benchmark for binary sentiment classification at the document or snippet level. Sourced from the critical review aggregator website, it consists of raw text snippets from professional movie reviews, each labeled with a broad positive or negative sentiment. In contrast to the phrase-level detail of SST-2, the RT dataset typically presents a more direct and practical sentiment analysis task: classifying the overall polarity of a given text segment. Its significance stems from its realism and simplicity, providing a robust testbed for a model's performance on a straightforward, real-world opinion mining task, and it is often used in direct comparison with other sentiment datasets to evaluate generalizability and robustness.

\begin{table}[!t]
	\centering
	\caption{Basic information of the dataset.}
	\label{tab:dataset}
		\resizebox{0.45\textwidth}{!}{
		\begin{tabularx}{\linewidth}{cccc}
			\toprule
			\textbf{Dataset} & \textbf{Min text len} & \textbf{Max text len} & \textbf{Ave. text len}\\
			\midrule
			CoLA & 2 & 42 & 7.7\\
			\midrule
			SST & 2 & 52 & 19.14\\
			\midrule
			RT & 1 & 59 & 20.99\\
			\bottomrule
		\end{tabularx}
			}
				\vspace{-0.3cm}
\end{table}

\section{Baselines}\label{baselines}

We compare our attack with the following methods, including the optimization-based method (i.e., LAMP \cite{LAMP}) and the embedding-based method (i.e., DAGER \cite{dager}) designed for GIAs in
FedLLMs. 

\para{LAMP.} LAMP \cite{LAMP} is a novel and highly effective reconstruction attack designed to extract private, client-held text data from the gradient updates exchanged during the federated learning of transformer-based language models. The core innovation of LAMP lies in its two-pronged attack strategy, which alternates between continuous and discrete optimization. The continuous phase refines a set of dummy token embeddings by minimizing the difference between their resulting gradients and the target client's true gradients. Crucially, to overcome local optima and guide the reconstruction towards coherent text, LAMP introduces a discrete optimization phase that uses an auxiliary language model, such as GPT-2, as a prior. This phase generates candidate sequences through textual transformations and selects the one that best balances gradient-matching loss with low linguistic perplexity. This synergistic approach enables LAMP to significantly outperform prior attacks, successfully reconstructing longer, more accurate text sequences, including in challenging scenarios with batch sizes larger than one and against defended models, thereby exposing a greater vulnerability to privacy leakage in federated learning systems than previously known.

\para{DAGER.} DAGER \cite{dager} is a groundbreaking gradient inversion attack that, for the first time, achieves the exact reconstruction of entire batches of private text input to large language models (LLMs) from their gradients. Departing from the approximate, optimization-based approaches of prior methods, DAGER leverages two key properties of transformer models: the low-rank structure of self-attention layer gradients and the discrete nature of tokenized text. Its core mechanism involves a highly efficient ``span check'' that verifies whether a candidate token's embedding lies within the column space of the observed gradient matrices. This allows DAGER to first recover the set of all tokens present in a client's batch and then, through a greedy algorithm for decoders or a heuristic search for encoders, reassemble them into the original sequences with perfect accuracy. This approach enables DAGER to scale to significantly larger batch sizes (up to 128) and longer sequences than previously possible, achieving near-perfect ROUGE scores and dramatically outperforming prior attacks in both reconstruction quality and computational speed.
\begin{figure*}[!t]
	\centering
	\includegraphics[width=0.6\textwidth]{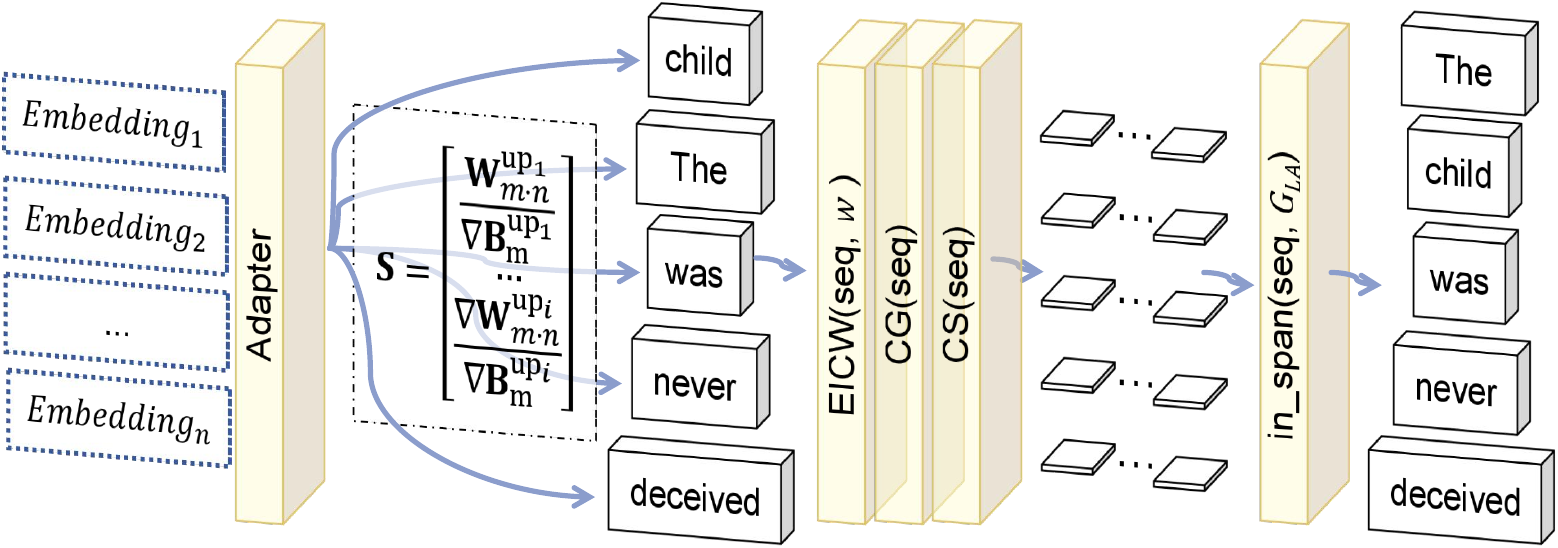}
	\caption{The showcase of the UTR attack.}\label{showcase}
\end{figure*}
 \begin{table*}[!t]
	\caption{An example of UTR on CoLA with batch size 2.}
	\label{T:utr}
	\centering
	\resizebox{\textwidth}{!}{
		
		\begin{tabular}{ccl}
			\hline 
			\multirow{2}{*}{Origin text} & \multicolumn{2}{l}{The child was never deceived.} \\
			& \multicolumn{2}{l}{Jason persuaded Medea that she should desert her family.} \\ \cline{1-3}
			
			\multirow{4}{*}{BERT} & \multirow{2}{*}{word bag} & \texttt{`\#\#ved',`should',`never',`desert',`dec',`med',`the',`was ',`child',`that',} \\
			& & \texttt{ `her', `she',`[CLS]',`[SEP]',`family',`jason',`.',`persuaded',`\#\#ei',`\#\#ea'} \\ \cline{2-3} 
			& \multirow{2}{*}{Recon. text} & \texttt{the child was never deceived. } \\ 
			& & \texttt{jason persuaded medea that she should desert her family.} \\ \hline
			
			\multirow{4}{*}{GPT} & \multirow{2}{*}{word bag} & \texttt{`.',`was ',`that',`The',`her',`she',`should',`child',`never',`family',`Med',} \\
			& & \texttt{`desert',`ea',`persuaded',`Jason',`deceived',`<|endoftext|>'} \\ \cline{2-3} 
			& \multirow{2}{*}{Recon. text} & \texttt{the child was never deceived. } \\ 
			& & \texttt{Jason persuaded Medea that she should desert her family.} \\ \hline
			
			\multirow{4}{*}{Qwen} & \multirow{2}{*}{word bag} & \texttt{`.', ` that', ` was ', `The', ` her', ` should', ` she', ` child', ` never', } \\
			& & \texttt{`ea', ` desert', `Jason', ` persuaded', ` deceived'} \\ \cline{2-3} 
			& \multirow{2}{*}{Recon. text} & \texttt{The child was never deceived. } \\ 
			& & \texttt{Jason persuaded Medea that she should desert her family.} \\ \hline
		\end{tabular}
	}
\end{table*}

\para{Discussion.} During experiments, we employed only the first two Transformer layers for two reasons. First, UTR requires gradient data from the embedding and first-layer adapters to evaluate privacy leakage, and these are fully captured within the initial layers. Second, due to computational constraints—specifically, a single NVIDIA RTX 4090 GPU (24 GB VRAM) on a Windows system—running the full model was infeasible. Thus, restricting to two layers provided an effective and practical balance between experimental validity and resource limitations.

\section{Attack Example of UTR}\label{example}

\para{Sample Reconstruction.} We show sample sequence reconstructions from UTR on CoLA with $B = 2$ in Table~\ref{T:utr}. We can observe that our reconstruction is coherent, which can successfully reconstruct the accurate original text and Fig. \ref{showcase} illustrates the showcase of our UTR.

\section{Additional Experimental Results}
\label{A1}

\para{Model Utility Evaluation under DP.} As illustrated in Table~\ref{tab:performance_high_dp}, when strong Differential Privacy (DP) guarantees are enforced, the model's utility drops to 69.20\%. This performance level remains substantially below what is considered sufficient for real-world deployment (80\%-82\%). 

\begin{table}[!t]
	\centering
	\caption{Model utility evaluation under DP.}\label{tab:performance_high_dp}
	\resizebox{0.4\textwidth}{!}{
		\begin{tabular}{cccc}
			\toprule
			\multirow{2}{*}{\textbf{Defense}} & \multirow{2}{*}{\textbf{Parameter}} & \multicolumn{2}{c}{BERT}  \\
			\cmidrule(r){3-4} 
		&  & Loss & Accuracy \\
		\midrule
		\multirow{2}{*}{\textbf{DP}} & $\sigma = 0.0001 $ & $\sim 0.42$ & $\sim 81.2\%$   \\
		& $\sigma = 10.0$ & $\sim 0.61$ & $\sim 69.2\%$ \\
			\bottomrule
	\end{tabular}}
\end{table}

\para{Attack Performance Evaluation under Structural Defense.} In response to the suggestion regarding structural defenses, we note that while such architectural modifications are generally avoided to prevent client-side heterogeneity, they offer valuable insights in an experimental context. Our supplementary tests reveal that architectural hardening provides measurable benefits: substituting ReLU with GELU lowered reconstruction by $\approx 15\%$, while utilizing deeper adapters (depth=3) reduced leakage by $\approx 10\%$ ($B < 64$) and $\approx 80\%$ ($B \geq 64$). These results align with and support our primary defense analysis.

\begin{table}[!t]
\centering
\caption{Computational resource overhead for data reconstruction across methods and tasks.}
\label{tab:compute_overhead}
\resizebox{0.4\textwidth}{!}{
\begin{tabular}{lccc}
\toprule
\textbf{Task} & \textbf{Method} & \textbf{$B$} & \textbf{GPU Time (s)} \\
\midrule
RT-GPT2            & UTR             & 128                 & 6,125 \\
RT-GPT2            & FILM            & 128                 & $\sim$26,000 \\
CoLA-GPT2     & UTR             & 64                  & 900 \\
CoLA-GPT2     & FILM            & 64                  & $\sim$54,000 \\
\bottomrule
\end{tabular}}
\end{table}

\para{Attack Efforts Evaluation.} We further evaluated the computational efficiency of UTR on open-source benchmarks. As summarized in Table~\ref{tab:compute_overhead}, UTR achieves significantly lower computational overhead compared to state-of-the-art methods. For instance, on the RT dataset ($B=128$), UTR requires only 6,125 GPU-seconds ($\sim$1.7 GPU-hours) to reconstruct \textit{all} batches, which is over 5× more efficient than FILM. Moreover, on the CoLA-GPT2 task ($B=64$), UTR completes reconstruction in approximately 15 minutes, yielding a speedup of more than $60\times$ relative to FILM. 

\end{document}